\documentclass[journal]{IEEEtran}

\usepackage[cmex10]{amsmath}
\usepackage{amssymb,amsfonts,bm,mathtools}
\usepackage{amsthm}
\DeclareMathAlphabet\mathbfcal{OMS}{cmsy}{b}{n}

\usepackage{graphicx,epstopdf,epsfig}
\usepackage{float}                 
\floatstyle{plaintop}\restylefloat{table}
\usepackage{booktabs,array,tabularx,multirow,makecell,hhline,relsize}
\usepackage{subcaption}
\usepackage{wrapfig,multicol,cuted}
\usepackage{perpage}\MakeSorted{figure}\MakeSorted{table}

\usepackage{algorithm}
\usepackage[noend]{algpseudocode}

\usepackage{tikz}\usetikzlibrary{positioning}
\usepackage{enumitem,etoolbox,lettrine,pifont}  
\usepackage[T1]{fontenc}
\usepackage{cite}
\usepackage[hidelinks]{hyperref}

\newtheorem{theorem}{Theorem}
\newtheorem{proposition}{Proposition}
\newtheorem{lemma}{Lemma}
\newtheorem{corollary}{Corollary}
\theoremstyle{remark}\newtheorem{remark}{Remark}

\makeatletter\def\BState{\State\hskip-\ALG@thistlm}\makeatother

\newcommand{\bH}{\mathbf{H}}
\newcommand{\bZ}{\mathbf{Z}}
\newcommand{\bF}{\mathbf{F}}
\newcommand{\bI}{\mathbf{I}}
\newcommand{\bR}{\mathbf{R}}
\newcommand{\bx}{\mathbf{x}}
\newcommand{\by}{\mathbf{y}}
\newcommand{\bn}{\mathbf{n}}
\newcommand{\bs}{\mathbf{s}}
\newcommand{\bz}{\mathbf{z}}

\newcommand{\bv}{\mathbf{v}}
\newcommand{\br}{\mathbf{r}}
\newcommand{\bPi}{\bm{\Pi}}
\newcommand{\bDelta}{\bm{\Delta}}
\newcommand{\bPsi}{\bm{\Psi}}
\newcommand{\bPhi}{\bm{\Phi}}
\newcommand{\bL}{\mathbf{L}}
\newcommand{\bE}{\mathbf{E}}
\newcommand{\bG}{\mathbf{G}}
\newcommand{\bB}{\mathbf{B}}
\newcommand{\bT}{\mathbf{T}}
\newcommand{\bC}{\mathbf{C}}
\newcommand{\bK}{\mathbf{K}}
\newcommand{\bg}{\mathbf{g}}
\newcommand{\bff}{\mathbf{f}}
\newcommand{\hmn}{H}
\newcommand{\Z}{\mathbb{Z}}
\begin{document}
	\title{Delay–Doppler Duality: Simplifying Exact Equalization of Zak-OTFS in Time or in Frequency}
	\author{Mudasir Ahmad Sheikh and Robert Calderbank
    \thanks{M. A. Sheikh and R. Calderbank are with the Department of Electrical and Computer Engineering, Duke University, Durham, NC, USA (e-mail: mudasir.sheikh@duke.edu; robert.calderbank@duke.edu). R. Calderbank is also with Cohere Technologies Inc., CA, USA. This work is supported by the National Science Foundation under grant 2342690 and by Cohere Technologies.}}
	\maketitle
	\IEEEpeerreviewmaketitle
	
\begin{abstract}
The Zak-OTFS carrier waveform is a pulse in the delay--Doppler (DD) domain and given the input--output (I/O) relation at one DD point in a frame, we can predict the I/O relation at all other points. While it is simple to acquire the effective wireless channel with Zak-OTFS, the cost of embracing inter-carrier interference is equalization complexity. The cost of linear minimum mean-square error (LMMSE) equalization grows cubically with the frame size. We describe how to use channel knowledge to reduce equalization complexity. For any pulse shaping filter, and for fractional delays and Dopplers, we show that the DD channel matrix is unitarily similar, through the discrete Zak transform, to a time-domain (TD) matrix that is a weighted sum of cyclic delay and Doppler operators, where the weights are the taps of the effective channel. We also obtain a dual frequency-domain (FD) form through an additional Fourier transform. The TD matrix is banded with a width set by the delay spread that is independent of the Doppler spread, whereas the FD matrix is banded with a width that is set only by the Doppler spread. We design a receiver that reads off the two widths from the channel estimate, selects the cheaper domain, and solves the banded system exactly, by a bordered Cholesky factorization that treats the quasi-periodic wrap-around without sacrificing degrees of freedom. Filters that decay slowly need a few preconditioned conjugate gradient steps. We show that the complexity of equalization becomes linear in the frame size. Simulations match direct DD-LMMSE error rates with three to four orders of magnitude fewer multiplications. A single receiver is able to handle channels with very different delay and Doppler characteristics.
\end{abstract}


\section{Introduction}
\IEEEPARstart{T}{he} challenge facing sixth generation (6G) wireless systems is reliable operation in a wide variety of high-mobility scenarios, including high-speed rail \cite{zhu2013tdd,hasegawa2018high,liu2019recent}, vehicle-to-everything (V2X) links, unmanned aerial vehicles, and low-Earth-orbit (LEO) satellites \cite{9849120,zhang2025leo}, where severe Doppler spreads and long delay spreads occur together \cite{wei2021orthogonal}. One approach is to continue with orthogonal frequency division multiplexing (OFDM), the waveform adopted in 4G/5G, and to develop custom signal processing chains for a variety of channel models \cite{6g_1}. We propose to start from an estimate of the effective channel (formed by transmit filter, physical channel, and receive filter), and to develop a single receiver that is able to handle a variety of channels.

Delay--Doppler (DD) methods (Zak-OTFS) \cite{book,bits22,bits23} and time--frequency methods (OFDM) offer different ways of balancing channel estimation and equalization. In OFDM, if the effective channel is known, then single-tap equalization is possible when inter-carrier interference (ICI) is commensurate with noise. However, in high-mobility scenarios, it is highly non-trivial to acquire the effective channel, and OFDM performance degrades severely \cite{wang}. 
In Zak-OTFS, the carrier waveform is a localized pulse in the DD domain, making it possible to separate reflectors along both delay and Doppler simultaneously, and to achieve a high-resolution DD radar image of the environment. Zak-OTFS \cite{book,bits22,bits23} removes this limitation by making each carrier a quasi-periodic pulse in the DD domain, parameterized by a delay period and a Doppler period whose product is unity. 
When the delay period is greater than the effective channel delay spread, and the Doppler period is greater than the effective Doppler spread, there is no aliasing in the DD domain and the Zak-OTFS input--output (I/O) relation is predictable and non-fading \cite{bits23,bits22}. Zak-OTFS requires no cyclic prefix (CP) \cite{bits23}, and we can read off the effective channel at every carrier of the frame from the channel response to a single DD pilot (possibly a spread pilot superimposed on data as in \cite{ubadah}). In Zak-OTFS it is straightforward to acquire the effective channel and we do not need to rely on the accuracy of a parametric channel model \cite{bits23,mattuTut}. The region of predictable operation can be widened further by interleaving pilots \cite{jayachandran2025interleaved}, and the spectral efficiency of the frame can be improved by mounting symbols on mutually unbiased bases \cite{mub2026}; both are orthogonal to the equalization problem studied here. In this paper we describe how to use knowledge of the effective channel to reduce equalization complexity.

We do not consider variants of Zak-OTFS, such as MC-OTFS \cite{hadani, raviteja2018interference, sheikh2024bayesian,sheikh2025RISOTFS}, that are built on top of time--frequency multicarrier modulation, because the effective channel is not \emph{predictable} when the Doppler spread is large \cite{bits23,mattuFD}. We note that pulse-shaping filters that trade time and bandwidth expansion against predictability have been characterized in \cite{jayachandran}. We also do not consider implementation, but we note that over-the-air prototypes with spread pilots have been reported at sub-6~GHz and sub-THz frequencies \cite{zheng,parisi}.

Zak-OTFS embraces inter-carrier interference and requires a receiver that performs joint equalization of all carriers. Linear minimum mean-square error (LMMSE) equalization in the DD domain requires inversion of a large matrix with no apparent structure. The size of the matrix is the number of carriers in a frame, and the cost of inversion is cubic in the frame size \cite{mattuFD}. For a typical frame of 31 delay bins and 37 Doppler bins, this amounts to about a billion complex multiplications per frame, far beyond real-time budgets. A recent tutorial \cite{mattuTut} highlights overcoming this cubic barrier as an important open problem. Low-complexity linear equalizers have been studied for MC-OTFS. 
Low-complexity linear equalizers have been studied for MC-OTFS. Assuming a rectangular pulse, a CP, and integer delays and Dopplers, Tiwari \emph{et al.} \cite{tiwari} observed that the time-domain (TD) channel matrix is quasi-banded, and implemented LMMSE equalization through a partitioned lower--upper (LU) factorization in time followed by fast Fourier transforms (FFTs). Surabhi and Chockalingam \cite{surabhi} exploited block circulant channel structure for the same purpose. Singh \emph{et al.} \cite{9933842} designed a low-complexity LMMSE receiver for MIMO-OTFS with practical pulse shaping, so that non-ideal filters are not by themselves an obstacle; the structure they exploit is still that of a CP-based multicarrier system with on-grid taps. For massive MIMO OTFS, Sheikh \emph{et al.} \cite{sheikh2022low} observed that the channel matrix is multi-banded with a large bandwidth, since its non-zero entries occur in several scattered bands, and reduced that bandwidth by reordering the matrix with the Reverse Cuthill--McKee algorithm before an exact MMSE solve. All these works neglect fractional delays and Dopplers, making it difficult to apply them to Zak-OTFS, where filters and off-grid paths spread the contribution of every path to the effective channel over many DD taps. All these rely on a CP and on ideal pulse shaping to create the cyclic structure that Zak-OTFS obtains directly from quasi-periodicity. Further, and all three assemble the system matrix from path parameters, whereas a Zak-OTFS receiver acquires the effective channel as a sampled DD function without \emph{a priori} knowledge of a path model \cite{bits23,ubadah}. Finally, the LU factorization given in \cite{tiwari} ignores that the LMMSE system matrix is Hermitian positive definite, and does not account for the phase that arises when the order of multiplying a delay and a Doppler operator is reversed.

Mohammed \emph{et al.} \cite{mohammedFD} revealed structure in the Zak-OTFS system matrix by showing that the frequency-domain (FD) and DD domain I/O relations are unitarily equivalent, and that the FD channel matrix is banded along Doppler. This established the key idea of equalizing Zak-OTFS in a unitarily equivalent banded domain. However, the LMMSE estimate in \cite{mohammedFD} is computed by banded matrix inversion, at a cost that remains quadratic in the frame size. The cyclic wrap-around of the FD matrix is handled by introducing an extended symbol vector with edge entries that duplicate symbols from the opposite edge; the system matrix is strictly banded only because the duplicated entries are treated as distinct. Also, their regularization assumes white noise \cite[eq.~(33)]{mohammedFD}, though matched filtering with a Gaussian pulse and window colors the noise, and we will assume the more accurate model in our system model. The extended version by Mattu \emph{et al.} \cite{mattuFD} evaluates estimated channels with several pulse shapes and uses the conjugate gradient (CG) method to reduce equalization complexity to linear in the frame size. The cyclic wrap-around is removed by forcing the first and last FD samples of the frame to zero, since those samples cannot be controlled directly. The information symbols are mounted on the null space of the corresponding transform, which confines the data to a subspace of reduced dimension and sacrifices degrees of freedom that grow with the band width. Being iterative, the method also has a cost that depends on the condition number of the LMMSE system matrix, which increases with the signal-to-noise ratio (SNR) and with poorly conditioned channels, so that the processing latency is not fixed in advance.

Equalization is moreover confined to the FD, and the width of the banding grows with the Doppler spread, so the method becomes costlier in exactly the high-mobility regime that motivates Zak-OTFS. The TD alternative is set aside on the grounds that the TD channel matrix is dense for fractional, frequency-selective channels \cite{mattuFD}. 
Zak-OTFS has also been equalized directly in the TD by Zhang et al. \cite{zhang2025ogd}, where they apply offset gradient descent to the TD system. We differ in three respects. Their receiver is iterative, so its accuracy and its cost are set by the number of descent steps, whereas we solve exactly with an operation count that is fixed in advance. Their descent treats the TD system as unstructured, whereas we show that it is banded, with a width set by the delay spread alone, and it is this structure that makes a direct solve affordable. And they equalize in the TD only, whereas we select the cheaper of the two domains, which is worth up to a factor of $32$ in the channels considered here. Iterative solvers remain useful for slowly decaying filters, and there our factorization serves as their preconditioner rather than as a competitor.

Three open questions remain. First, there is no unified view of Zak-OTFS equalization in the TD and FD that covers practical filters and fractional channels, and no rule for choosing between the two domains. Second, there is no exact direct solver with linear complexity that represents the quasi-periodic wrap-around without sacrificing degrees of freedom or approximation, and for which the operation count is fixed in advance. Third, there is no systematic treatment either of slowly decaying filters, for which any banded truncation leaves a residual error floor, or of the noise coloring introduced by the receive filter. We answer these questions and summarize our main contributions as follows.
\begin{itemize}
\item Starting from the continuous-time system model with general pulse shaping filters, we derive the discrete-time Zak-OTFS channel matrix in closed form, without assuming integer taps. This matrix is governed by the cross-ambiguity function of the transmit pulse, which we evaluate in closed form for the Gaussian and sinc pulses. We describe how decay of the transmit pulse controls the width of the banding in the TD and FD representations.
\item We prove that the DD channel matrix is unitarily similar to a weighted sum of cyclic delay and Doppler operators, where the weights are the taps of the effective channel, and we show that one further Fourier transform exchanges the roles of delay and Doppler. The FD relation of \cite{mohammedFD} and the integer-tap multicarrier model of \cite{tiwari} are special cases of this single identity, and the correspondence between the channel taps and the matrix is one to one and computable with FFTs.
\item We show that the TD banding width is set by the delay spread alone, and the FD banding width by the Doppler spread alone. It follows that time-domain equalization is dense for slowly decaying filters but not in general. Since the energy on each diagonal is a marginal of the effective channel power, our receiver is able to read off both banding widths from the effective channel estimate in $\mathcal O(MN)$ operations and equalize in whichever domain is cheaper. We demonstrate a factor of 32 improvement for the channels we consider in this paper.
\item Our receiver selects the cheaper domain and solves the resulting cyclically banded system exactly by a bordered banded Cholesky factorization that confines every wrap-around entry. The computational cost is fixed in advance, linear in the frame size and quadratic in the number of retained diagonals, with no extended vector, no null-space mounting and no iterations. We show that noise coloring by the receive filter preserves the banding.
\item We show that with the sinc filter the effective channel decays only as the reciprocal of the delay index, so no affordable banding captures the effective channel, and truncation alone leaves an error floor. We remove that floor in two steps. First, a regularization that charges the discarded energy to the noise term, which costs nothing. Second, a few conjugate gradient iterations on the exact operator, preconditioned by the banded factorization already computed. Four such iterations reach the exact solution whereas CG without preconditioning needs sixteen.
\item We demonstrate that, with Gaussian filtering, the proposed equalizer achieves the same performance as direct DD-LMMSE while requiring three to four orders of magnitude fewer complex multiplications. For noisy channel estimates, the retained-band structure also suppresses weak, noise-dominated channel components, providing improved robustness over direct equalization.
\end{itemize}
There is significant recent interest in designing new modulation schemes for doubly-selective channels with large delay and Doppler spreads, where legacy modulation schemes based on time-frequency signal representations underperform. Multiple modulation schemes, e.g., in the DD, chirp, time-sequency, and other domains, have been proposed in the literature for this purpose, with varying implementation details. We have established that all previously proposed modulation schemes for doubly-selective signaling are instances of a single family of complex Hadamard-modulated pulse trains \cite{MMC}. We plan to extend the equalization framework developed in this paper to these modulations in subsequent work. 
	\section{System Model}\label{sec:model}
Let $\tau_p$ and $\nu_p=1/\tau_p$ be the delay and Doppler periods. The fundamental period is sampled with $M$ delay bins and $N$ Doppler bins, so the bin widths are $\tau_p/M=1/B$ and $\nu_p/N=1/T$, with bandwidth $B=M\nu_p$ and duration $T=N\tau_p$. There are $n\triangleq MN=BT$ symbols $x[k,l]$, $k\in\{0,\dots,M-1\}$, $l\in\{0,\dots,N-1\}$, drawn independent and identically distributed (i.i.d.) from a unit-energy constellation. 
\subsection{Transmitter}
At the transmitter, the symbols $x[k,l]$ are embedded into a
quasi-periodic discrete DD-domain signal, since the corresponding
time-domain (TD) realization is defined only for quasi-periodic signals
\cite{mohammedFD, bits22}, as
\begin{equation}\label{eq:xdd}
	x_{\mathrm{dd}}[k,l]
	=x[k\bmod M,\,l\bmod N]\,
	e^{j2\pi\lfloor k/M\rfloor l/N}.
\end{equation}
The resulting signal satisfies the quasi-periodicity condition with
periods $M$ and $N$ along the discrete delay and Doppler axes,
respectively. Specifically, for all $k,l,n,m\in\mathbb{Z}$,
\begin{equation}\label{eq:qp}
	x_{\mathrm{dd}}[k+nM,\,l+mN]
	=e^{j2\pi n l/N}x_{\mathrm{dd}}[k,l].
\end{equation}
The inverse discrete Zak transform (IDZT) maps $x_{dd}$ to a discrete-time sequence \cite{book,mohammedFD},
\begin{equation}\label{eq:idzt}
	s[k+qM]=\frac{1}{\sqrt N}\sum_{l=0}^{N-1}x_{dd}[k,l]\,e^{j2\pi\frac{ql}{N}},\qquad k\in\Z,\;q\in\Z.
\end{equation}
This defines $s[t]$ for every $t\in\Z$, because every integer can be written as $t=k+qM$. Pulse shaping with a TD pulse $w_1(t)$ and a TD window $W_2(t)$ gives the transmit signal \cite{mohammedFD}
\begin{align}
	x(t)&=\sqrt T\,w_1(t)\star\Big[W_2(t)\sum_{m\in\Z}s[m]\,\delta(t-m/B)\Big]\nonumber\\
	&=\sqrt T\sum_{m\in\Z}s[m]\,W_2(m/B)\,w_1(t-m/B).\label{eq:xt}
\end{align}
The window $W_2$ confines the signal to about $T$ seconds and the pulse $w_1$ confines it to about $B$ hertz. This pair of operations is exactly the time and frequency windowing of Zak-OTFS analyzed in \cite{gopalam2024windowing}. Equivalently, this is twisted convolution of $x_{dd}(\tau,\nu)$ with $w_{tx}(\tau,\nu)=w_1(\tau)w_2(\nu)$, where $w_2$ is the Fourier transform of $W_2$ \cite{mohammedFD}. Both filters have unit energy.
\subsection{Channel}
The physical channel consists of $P$ propagation paths,
\begin{equation}\label{eq:hphy}
	h_{\rm phy}(\tau,\nu)=\sum_{i=1}^{P}h_i\,\delta(\tau-\tau_i)\,\delta(\nu-\nu_i),
\end{equation}
whose delays $\tau_i$ and Dopplers $\nu_i$ are not constrained to the DD grid. The received signal is
\begin{equation}\label{eq:rt}
	r(t)=\sum_{i=1}^{P}h_i\,x(t-\tau_i)\,e^{j2\pi\nu_i(t-\tau_i)}+n(t),
\end{equation}
where $n(t)$ is white Gaussian noise with power spectral density $N_0$.

\subsection{Receiver}
The receiver matched-filters with $w_1$, applies the window, samples at $t/B$ and periodizes with period $n$ \cite{mohammedFD}:
\begin{align}
	y(t)&=\frac{1}{\sqrt T}W_2^*(t)\int w_1^*(u-t)\,r(u)\,du,\label{eq:yt}\\
	y_T[t]&=\sum_{p\in\Z}y\!\left(\frac{t+pn}{B}\right),\qquad t\in\{0,\dots,n-1\}.\label{eq:yT}
\end{align}
The discrete Zak transform (DZT) of $y_T$ gives the received DD samples
\begin{equation}\label{eq:dzt}
	y_{dd}[k,l]=\frac{1}{\sqrt N}\sum_{q=0}^{N-1}y_T[k+qM]\,e^{-j2\pi\frac{ql}{N}}.
\end{equation}

\subsection{DD Input--Output Relation}
Combining the transmit, channel, and receive operations in \eqref{eq:idzt}--\eqref{eq:yT}, followed by the DZT in \eqref{eq:dzt}, yields the DD-domain
input--output relation
\begin{equation}\label{eq:io}
	\begin{split}
		y_{dd}[k,l]=\sum_{k',l'\in\Z}&h_{\rm eff}[k',l']\,x_{dd}[k-k',l-l']\\
		&\times e^{j2\pi\frac{l'(k-k')}{n}}+n_{dd}[k,l],
	\end{split}
\end{equation}
where $h_{\rm eff}[k',l']=h_{\rm eff}(k'/B,l'/T)$ samples the effective channel $h_{\rm eff}=w_{rx}*_\sigma h_{\rm phy}*_\sigma w_{tx}$, and $*_\sigma$ is twisted convolution \cite{book}. Because of \eqref{eq:qp}, \eqref{eq:io} can be written as an $n$-periodic twisted convolution. It uses the periodized channel
\begin{equation}\label{eq:hper}
	h[k',l']=\sum_{\bar p,\bar q\in\Z}h_{\rm eff}[k'+\bar pn,\,l'+\bar qn],
\end{equation}
in place of $h_{\rm eff}$, with $k',l'$ summed over one period \cite{mohammedFD}. Stacking the $MN$ DD-domain samples according to the linear index
$i=kN+l$ gives the equivalent vector model \cite{raviteja2018interference}
\begin{equation}\label{eq:yHx}
	\by=\bH\bx+\bn,
\end{equation}
where $\bR$ captures noise coloring by the receive filter. Note that $ \mathbb E[\bx\bx^\hmn]=\bI,\;\mathbb E[\bn\bn^\hmn]=N_0\bR$.
\subsection{LMMSE Equalization and Its Direct Cost}
The LMMSE estimator follows from the orthogonality principle, $\mathbb E[(\bx-\mathbf W\by)\by^\hmn]=\mathbf 0$. This gives $\mathbf W=\mathbb E[\bx\by^\hmn](\mathbb E[\by\by^\hmn])^{-1}$ with $\mathbb E[\bx\by^\hmn]=\bH^\hmn$ and $\mathbb E[\by\by^\hmn]=\bH\bH^\hmn+N_0\bR$. Hence
\begin{equation}\label{eq:lmmse}
	\hat\bx=\bH^\hmn\left(\bH\bH^\hmn+N_0\bR\right)^{-1}\by.
\end{equation}
Forming $\bH\bH^\hmn$ costs $n^3/2$ complex multiplications (CMs), using Hermitian symmetry. The Cholesky factorization costs $n^3/6$ and the solves $\mathcal O(n^2)$. The direct cost is therefore about $\tfrac23 n^3$ CMs. To reduce this computational cost, first, the cost of solving \eqref{eq:lmmse} is not set by the amount of interference, which is fixed by the channel, but by how that interference is \emph{arranged} in the matrix: a system matrix whose nonzero entries lie close to the diagonal is factorized in a cost linear in its order, whereas a matrix with the same entries scattered is not \cite{golub}. Second, the DD symbol vector is only one of several equivalent descriptions of the same frame. The DZT \eqref{eq:dzt} relates it to the time-domain samples, and a further Fourier transform to the frequency-domain samples, both unitarily, so the LMMSE estimate is unchanged by working in either of them. Section~\ref{sec:dual} shows that these two descriptions arrange the interference far more favorably than the DD one, and quantifies how favorably in terms of the effective channel alone.

\section{Exact Discrete-Time Model and Delay--Doppler Duality}\label{sec:dual}
\subsection{Exact Time-Domain Channel}
\begin{proposition}[TD channel]\label{prop:td}
	Let $\hat W\triangleq\sqrt T\,W_2$, and let $t_p\triangleq(t+pn)/B$ and $t'_q\triangleq(t'+qn)/B$. Up to a common scale factor $1/T$ on signal and noise, the periodized samples \eqref{eq:yT} satisfy $\by_T=\bH_T\bs+\bn_T$, where $\bs=[s[0],\dots,s[n-1]]^{\mathsf T}$ and
	\begin{equation}\label{eq:HTphys}
		\begin{split}
			\bH_T[t,t']=\sum_{p,q\in\Z}&\hat W^*(t_p)\,\hat W(t'_q)\\
			&\times\sum_{i=1}^{P}h_i\,\chi_{\nu_i}(t_p-t'_q-\tau_i)\,e^{j\pi\nu_i(t_p+t'_q-\tau_i)}.
		\end{split}
	\end{equation}
	Here $\chi_\nu$ is the symmetric cross-ambiguity function of the pulse,
	\begin{equation}\label{eq:chi}
		\chi_\nu(d)\triangleq\int w_1^*\!\left(v-\tfrac d2\right)w_1\!\left(v+\tfrac d2\right)e^{j2\pi\nu v}\,dv.
	\end{equation}
\end{proposition}
\begin{IEEEproof}
	Substituting \eqref{eq:xt} into \eqref{eq:rt} gives
	\begin{multline*}
		r(t)=\sqrt T\sum_i h_i\sum_m s[m]W_2(m/B)\\
		\times w_1(t-\tau_i-m/B)\,e^{j2\pi\nu_i(t-\tau_i)}+n(t).
	\end{multline*}
	Inserting this into \eqref{eq:yt}, the signal part of $y(t)$ is
	\[
	W_2^*(t)\sum_i h_i\sum_m s[m]W_2(m/B)\,I_i(t,m),
	\]
	where
	\[
	I_i(t,m)=\int w_1^*(u-\alpha)\,w_1(u-\beta)\,e^{j2\pi\nu_i(u-\tau_i)}\,du,
	\]
	with $\alpha=t$ and $\beta=\tau_i+m/B$.
	
	The integral is a cross-ambiguity in disguise. Substitute $u=\frac{\alpha+\beta}{2}+v$ and set $d=\alpha-\beta$. Then $u-\alpha=v-\frac d2$ and $u-\beta=v+\frac d2$, so
	\[
	I_i=e^{j2\pi\nu_i\left(\frac{\alpha+\beta}{2}-\tau_i\right)}\chi_{\nu_i}(d).
	\]
	Since $\frac{\alpha+\beta}{2}-\tau_i=\frac12\big(t+\frac mB-\tau_i\big)$ and $d=t-\frac mB-\tau_i$, the phase and the argument of $\chi$ take the form in \eqref{eq:HTphys}. To sample and periodize, evaluate at $t=(t_0+pn)/B$ and sum over $p$ as in \eqref{eq:yT}. The sequence $s$ is $n$-periodic by \eqref{eq:qpz}, so writing $m=t'+qn$ with $t'\in\{0,\dots,n-1\}$ gives $s[m]=s[t']$. Collecting the coefficient of $s[t']$ yields \eqref{eq:HTphys}, with $W_2^*W_2=\hat W^*\hat W/T$. The factor $1/T$ also multiplies the noise standard deviation, as shown in Section~\ref{sec:lmmse}, so it does not change the SNR and is dropped.
\end{IEEEproof}

\emph{Gaussian pulse.} For $w_1(t)=(2\alpha/\pi)^{1/4}e^{-\alpha t^2}$ with $\alpha=\alpha_g(B+2\nu_{\max})^2$ \cite{mohammedFD}, completing the square in \eqref{eq:chi} gives
\[
(v-\tfrac d2)^2+(v+\tfrac d2)^2=2v^2+\tfrac{d^2}{2}.
\]
The Gaussian integral $\int e^{-2\alpha v^2}e^{j2\pi\nu v}dv=\sqrt{\pi/(2\alpha)}\,e^{-\pi^2\nu^2/(2\alpha)}$ then gives
$\chi_\nu(d)=e^{-\alpha d^2/2}\,e^{-\pi^2\nu^2/(2\alpha)}.$
The window is $W_2(t)=\left(\frac{2}{\pi\beta_g(T+\tau_{\max})^2}\right)^{1/4}e^{-t^2/(\beta_g(T+\tau_{\max})^2)}$ \cite{mohammedFD}.

\emph{Sinc pulse.} For $w_1(t)=\sqrt B\,\mathrm{sinc}(Bt)$, the spectrum is $W_1(f)=B^{-1/2}\mathbf 1\{|f|<B/2\}$. Apply Parseval to \eqref{eq:chi}, with the spectra $W_1(f)e^{-j\pi fd}$ and $W_1(f-\nu)e^{j\pi(f-\nu)d}$. The spectra overlap on an interval of length $B-|\nu|$ centered at $\nu/2$, so
$\chi_\nu(d)=\frac{B-|\nu|}{B}\,\mathrm{sinc}\big((B-|\nu|)d\big),\qquad|\nu|<B.$
With the rectangular window $W_2(t)=T^{-1/2}\mathbf 1\{t\in[-\tfrac{1}{2B},T-\tfrac{1}{2B})\}$, only $p=q=0$ contributes to \eqref{eq:HTphys}. The Gaussian $\chi$ decays as $e^{-\alpha d^2/2}$, whereas the sinc $\chi$ decays only as $1/|d|$. 

The Gaussian pulse therefore confines the effective channel to a few delay taps, whereas the sinc tails decay too slowly for any narrow band to capture them. This is illustrated in Section~\ref{sec:results} which shows that with the Gaussian pulse the truncated solver is exact for practical purposes, whereas with the sinc pulse it leaves an error floor that only the refinement removes.

\subsection{Discrete Zak Transform}
Written as $\bx=\bZ\bs$, the DZT is a permutation followed by $M$ unitary $N$-point DFTs, $\bZ=(\bI_M\otimes\bF_N)\mathbf P$, and is therefore unitary. It also matches the two periodicities: if $s$ is $n$-periodic, then extending \eqref{eq:dzt} to all $k,l\in\Z$ gives a signal that is quasi-periodic in the sense of \eqref{eq:qp}, and conversely the IDZT \eqref{eq:idzt} of a quasi-periodic $x_{dd}$ is $n$-periodic. Indeed, periodicity in $l$ is immediate, while along delay
\begin{align}
	x_{dd}[k+M,l]&=\tfrac{1}{\sqrt N}\sum_{q=0}^{N-1}s[k+(q+1)M]e^{-j2\pi\frac{ql}{N}}\nonumber\\
	&=e^{j2\pi\frac lN}\tfrac{1}{\sqrt N}\sum_{q'=1}^{N}s[k+q'M]e^{-j2\pi\frac{q'l}{N}}\nonumber\\
	&=e^{j2\pi\frac lN}x_{dd}[k,l],\label{eq:qpz}
\end{align}
where the last step uses $s[k+NM]=s[k]$ and $e^{-j2\pi Nl/N}=1$; iterating gives \eqref{eq:qp}. Conversely, $s[k+(q+N)M]=\frac{1}{\sqrt N}\sum_l x_{dd}[k,l]e^{j2\pi(q+N)l/N}=s[k+qM]$ by \eqref{eq:idzt} and \eqref{eq:qp}. Quasi-periodicity in DD and $n$-periodicity in time are thus two views of the same property, and this is what creates the cyclic structure exploited in Section~\ref{sec:dual}, in place of the CP used in MC-OTFS.

With $\bx=\bZ\bs$ and $\by=\bZ\by_T$, Proposition~\ref{prop:td} gives $\bH=\bZ\bH_T\bZ^\hmn$. The next theorem expresses $\bH_T$ directly in terms of the DD channel $h[k',l']$.

\subsection{Duality Theorem}
The proof uses three elementary identities. First, $(\bDelta\bPi\bs)[t]=\omega^t s[t-1]$ while $(\bPi\bDelta\bs)[t]=\omega^{t-1}s[t-1]$, so $\bDelta\bPi=\omega\bPi\bDelta$ and, by iteration,
\begin{equation}\label{eq:comm}
	\bDelta^{l'}\bPi^{k'}=\omega^{k'l'}\,\bPi^{k'}\bDelta^{l'}.
\end{equation}
Second, $(\bF\bPi\bs)[f]=n^{-1/2}\sum_t s[t-1]\omega^{-ft}=\omega^{-f}(\bF\bs)[f]$ and $(\bF\bDelta\bs)[f]=(\bF\bs)[f-1]$, so
\begin{equation}\label{eq:fourierops}
	\bF\bPi\bF^\hmn=\bDelta^{-1},\qquad \bF\bDelta\bF^\hmn=\bPi.
\end{equation}
Third, $\bPi^{k'}\bDelta^{l'}$ has entries $\omega^{l'(t-k')}$ at positions $(t,t-k')$, so two such matrices with different delay indices have disjoint supports, and for equal delay indices the trace of the product is $\sum_t\omega^{(l''-l')(t-k')}=n\,\delta_{l'l''}$. Hence
\begin{equation}\label{eq:basis}
	\mathrm{tr}\big((\bPi^{k'}\bDelta^{l'})^\hmn\bPi^{k''}\bDelta^{l''}\big)=n\,\delta_{k'k''}\delta_{l'l''},
\end{equation}
and the $n^2$ matrices $\{\bPi^{k'}\bDelta^{l'}\}$ form an orthogonal basis of $\mathbb C^{n\times n}$.

\begin{theorem}[Delay--Doppler duality]\label{thm:dual}
	For every $n$-periodic effective channel $h[k',l']$ in \eqref{eq:hper}, the DD matrix in \eqref{eq:yHx} satisfies
	\begin{align}
		\bH&=\bZ\bH_T\bZ^\hmn,&\bH_T&=\sum_{k'=0}^{n-1}\sum_{l'=0}^{n-1}h[k',l']\,\bPi^{k'}\bDelta^{l'},\label{eq:HT}\\
		\bH_F&\triangleq\bF\bH_T\bF^\hmn,&\bH_F&=\sum_{k'=0}^{n-1}\sum_{l'=0}^{n-1}h[k',l']\,\bDelta^{-k'}\bPi^{l'}.\label{eq:HF}
	\end{align}
	Entrywise, $\bH_T[t,t-k']=g_{k'}[t]$ and $\bH_F[i,i-l']=f_{l'}[i]$, where
	\begin{equation}\label{eq:diags}
		g_{k'}[t]=\sum_{l'=0}^{n-1}h[k',l']\,\omega^{l'(t-k')},\;\; f_{l'}[i]=\sum_{k'=0}^{n-1}h[k',l']\,\omega^{-ik'}.
	\end{equation}
	The map $h\mapsto\bH_T$ is a bijection from $\mathbb C^{n\times n}$ onto $\mathbb C^{n\times n}$, with inverse
	\begin{equation}\label{eq:inv}
		h[k',l']=\frac1n\sum_{t=0}^{n-1}g_{k'}[t]\,\omega^{-l'(t-k')}.
	\end{equation}
\end{theorem}
\begin{IEEEproof}
	Consider first a single tap. Let $\bs=\bZ^\hmn\bx$, which is $n$-periodic by \eqref{eq:qpz}. Fix a tap $(k',l')$ and define
	\[
	r[t]\triangleq\omega^{l'(t-k')}s[t-k']=(\bPi^{k'}\bDelta^{l'}\bs)[t].
	\]
	This sequence is also $n$-periodic.	
	Its DZT is, by \eqref{eq:dzt},
	\begin{align*}
		(\bZ\br)[k,l]&=\tfrac{1}{\sqrt N}\sum_{q}s[k-k'+qM]\,\omega^{l'(k-k'+qM)}e^{-j2\pi\frac{ql}{N}}\\
		&\overset{(a)}{=}\omega^{l'(k-k')}\tfrac{1}{\sqrt N}\sum_q s[(k-k')+qM]\,e^{-j2\pi\frac{q(l-l')}{N}}\\
		&\overset{(b)}{=}e^{j2\pi\frac{l'(k-k')}{n}}\,x_{dd}[k-k',\,l-l'].
	\end{align*}
	Here $(a)$ uses $\omega^{l'qM}=e^{j2\pi l'q/N}$, and $(b)$ uses the DZT of $s$ evaluated at $(k-k',l-l')$, which by \eqref{eq:qpz} is the quasi-periodic extension $x_{dd}$, even when $k-k'\notin\{0,\dots,M-1\}$.
	This is exactly the $(k',l')$ term of the $n$-periodic form of \eqref{eq:io}. Summing over $(k',l')$ with weights $h[k',l']$ shows $\by-\bn=\bZ\big(\sum h[k',l']\bPi^{k'}\bDelta^{l'}\big)\bs=\bZ\bH_T\bZ^\hmn\bx$ for every $\bx$, which proves \eqref{eq:HT}.
	For the FD form, identity \eqref{eq:fourierops} gives $\bF\bPi^{k'}\bDelta^{l'}\bF^\hmn=\bDelta^{-k'}\bPi^{l'}$, which proves \eqref{eq:HF}.
	Reading off the diagonals, $\bPi^{k'}\bDelta^{l'}$ has entry $\omega^{l'(t-k')}$ at $(t,t-k')$, and $\bDelta^{-k'}\bPi^{l'}$ has entry $\omega^{-ik'}$ at $(i,i-l')$. This gives \eqref{eq:diags}.
	Finally, by \eqref{eq:basis}, the $n^2$ matrices $\{\bPi^{k'}\bDelta^{l'}\}$ form a basis of $\mathbb C^{n\times n}$. Taking the inner product of $\bH_T$ with $\bPi^{k'}\bDelta^{l'}$ gives \eqref{eq:inv}.
\end{IEEEproof}

\begin{remark}
	Proposition~\ref{prop:td} and Theorem~\ref{thm:dual} describe the same matrix. Proposition~\ref{prop:td} builds $\bH_T$ from path parameters, whereas \eqref{eq:diags} builds it from the DD channel estimate that a Zak-OTFS receiver acquires \cite{bits23,ubadah}. If $h[k',l']$ is supported on $D_T$ delay taps and $D_\nu$ Doppler taps, each $\bg_{k'}$ is a zero-padded $n$-point inverse FFT (IFFT) of the Doppler profile $h[k',\cdot]$, followed by a modulation. The TD channel therefore costs $D_T\,\frac n2\log_2 n$ CMs, and the FD channel $D_F\,\frac n2\log_2 n$. The $f_b$ in \eqref{eq:diags} coincides with the FD channel $h_f$ of \cite[Thm.~1]{mohammedFD}; Theorem~\ref{thm:dual} adds its TD dual and the unifying operator form.
\end{remark}


\subsection{Structural Consequences}
\begin{corollary}[Banded structure]\label{cor:band}
	If $h[k',l']=0$ for all $k'\notin\mathcal K$, where $\mathcal K$ is a set of $D_T$ cyclically consecutive delay taps, then $\bH_T$ has nonzero entries only on the $D_T$ cyclic diagonals $\{k'\in\mathcal K\}$. This holds for \emph{any} Doppler support. Dually, if $h[k',l']=0$ for $l'\notin\mathcal L$ with $|\mathcal L|=D_F$, then $\bH_F$ has $D_F$ cyclic diagonals, for any delay support. In the delay-major DD ordering, by contrast, $\bH$ is block-banded: its $N\times N$ blocks are nonzero only for $D_T$ block offsets, and each block is cyclically banded along Doppler with $D_F$ diagonals. The \emph{scalar} half-bandwidth of $\bH$ is therefore about $N D_T$.
\end{corollary}
\begin{IEEEproof}
	The TD and FD statements follow from \eqref{eq:diags}. For the DD statement, the delay-major index $kN+l$ groups rows and columns by delay bin. By \eqref{eq:io}, the entry at row $kN+l$ and column $k'N+l'$ is nonzero only when $(k-k')\bmod M\in\mathcal K\bmod M$ and $(l-l')\bmod N\in\mathcal L\bmod N$. The index offset between row and column is $(k-k')N+(l-l')$, whose magnitude reaches about $N D_T$.
\end{IEEEproof}

\noindent Corollary~\ref{cor:band} explains why the DD domain is the wrong place to factorize. A banded Cholesky factorization of an $n\times n$ matrix with half-bandwidth $w$ costs $\mathcal O(nw^2)$ \cite{golub}. With $w\approx ND_T$ in DD, this is $\mathcal O(nN^2D_T^2)$. In TD, the same channel has $w\approx D_T$, which is $N^2$ times cheaper. The DZT does not remove any interference; it reorders the problem so that interference only couples nearby indices.

\begin{lemma}[Band energy profiles]\label{lem:energy}
	Let $\bg_{k'}$ and $\bff_{l'}$ be the cyclic diagonals of $\bH_T$ and $\bH_F$ defined in \eqref{eq:diags}, and let $\mathbf h_{k'}\triangleq h[k',\cdot]$ and $\mathbf h^{\,\prime}_{l'}\triangleq h[\cdot,l']$ be the Doppler profile of delay tap $k'$ and the delay profile of Doppler tap $l'$. Then the energy carried by each diagonal is
	\begin{equation}\label{eq:energy}
		\|\bg_{k'}\|^2=n\,\|\mathbf h_{k'}\|^2,\qquad\|\bff_{l'}\|^2=n\,\|\mathbf h^{\,\prime}_{l'}\|^2,
	\end{equation}
	that is, $n$ times the delay marginal and $n$ times the Doppler marginal of the effective channel power $|h[k',l']|^2$. Summing \eqref{eq:energy} gives $\|\bH_T\|_F^2=\|\bH_F\|_F^2=n\sum_{k',l'}|h[k',l']|^2$, so the two domains distribute the same total energy over their diagonals, and only its \emph{concentration} differs.
\end{lemma}
\begin{IEEEproof}
	By \eqref{eq:diags}, $g_{k'}[t]=\sum_{l'}\tilde h_{k'}[l']\,\omega^{l't}$ with $\tilde h_{k'}[l']\triangleq h[k',l']\,\omega^{-k'l'}$, so $\tilde{\mathbf h}_{k'}$ is $\mathbf h_{k'}$ multiplied by a unimodular sequence and $\|\tilde{\mathbf h}_{k'}\|=\|\mathbf h_{k'}\|$. Hence $\bg_{k'}=\sqrt n\,\bF^\hmn\tilde{\mathbf h}_{k'}$, and Parseval gives $\|\bg_{k'}\|^2=n\|\tilde{\mathbf h}_{k'}\|^2=n\|\mathbf h_{k'}\|^2$. The FD case is identical with the roles of $k'$ and $l'$ exchanged. The last claim follows by summing over $k'$ and over $l'$, respectively, and noting that both sums equal the squared Frobenius norm.
\end{IEEEproof}
\begin{figure}[t]\centering
	\includegraphics[width=\columnwidth]{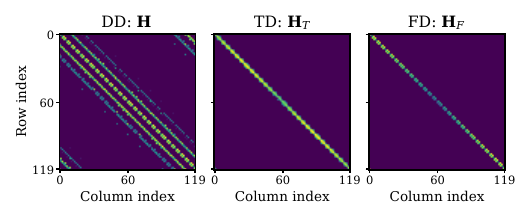}
	\caption{Magnitude of one Zak-OTFS channel in the DD, TD and FD domains ($M=12$, $N=10$, $n=120$, Veh-A, $\nu_{\max}=3$~kHz, Gaussian filter). }\label{fig:structure}
\end{figure}

Lemma~\ref{lem:energy} has a direct practical interpretation. The two
marginals of $|h[k',l']|^2$ can be obtained directly from the estimated
DD-domain channel and quantify the channel energy captured by a given
bandwidth in the TD and FD representations. Since these marginals
require only $\mathcal{O}(n)$ operations, they can be evaluated before
constructing the LMMSE matrix and used to select the processing domain
as described in Section~\ref{sec:trunc}. Fig.~\ref{fig:structure}
illustrates the corresponding matrix structures for one channel
realization.
For $M=12$, $N=10$, $B=360$~kHz and $T=0.33$~ms, the delay and
Doppler spreads are $B\tau_{\max}=0.9$ and
$2T\nu_{\max}=2$ bins, respectively. The three matrices contain the
same channel-induced interference, but its arrangement differs across
the domains, which directly determines the cost of equalization.

\begin{itemize}
	\item \emph{DD domain (left):} The channel matrix exhibits a
	two-level block-banded structure. Nonzero $N\times N$ blocks occur
	over five block offsets, $-2,\ldots,2$, while each block contains
	energy on a small number of cyclic Doppler diagonals. The
	quasi-periodicity also produces wrap-around blocks at the matrix
	boundaries. Consequently, the effective scalar half-bandwidth is
	approximately $ND_T=50$, illustrating why direct DD-domain
	factorization is considerably more expensive
	(Corollary~\ref{cor:band}).
	
	\item \emph{TD domain (center):} The energy is concentrated on five
	cyclic diagonals, with offsets $-2,\ldots,2$ at the $-50$~dB level.
	Thus, the scalar half-bandwidth is reduced from approximately $50$
	in the DD domain to $2$. The small corner entries represent the
	quasi-periodic wrap-around terms; although they contain only about
	$10^{-3}$ of the total energy, retaining them is necessary for an
	exact solution and is handled by the bordered factorization in
	\eqref{eq:part}.
	
	\item \emph{FD domain (right):} The energy is concentrated on three
	cyclic diagonals, with offsets $-1,0,1$, consistent with
	$2T\nu_{\max}+1=3$ from \eqref{eq:Dheur}. In this representation,
	the effective bandwidth is determined by the Doppler spread,
	independently of the delay spread.
\end{itemize}
For this realization, both the TD and FD representations substantially
reduce the scalar bandwidth relative to the DD representation, with the
FD representation being more compact because the Doppler spread is
smaller in normalized bins. At higher Doppler spreads, however, the
relative bandwidths can reverse, motivating the adaptive domain
selection rule in Section~\ref{sec:trunc}.
\section{Proposed Low-Complexity LMMSE Receiver}\label{sec:rx}
The same LMMSE problem can be solved in the DD domain, where the system matrix is full, or in either of two domains where it is banded. The receiver proposed here makes that choice per channel realization, and four principles guide its design. The change of domain must be exact, so that no loss is incurred beyond a deliberate and controlled band truncation. The domain and the retained band must be chosen per channel realization, from the DD channel estimate that a Zak-OTFS receiver already has. The resulting cyclically banded system must be solved by a direct factorization that treats the quasi-periodic wrap-around exactly and has a fixed operation count. Finally, iterations should be added only when the filter tails are too heavy for a narrow band, as with the sinc filter, whose effective channel decays only as the reciprocal of the delay index. Fig.~\ref{fig:rx} shows the resulting architecture. The \emph{channel path} (bottom row) depends only on the channel estimate and the noise level $N_0$; it selects the domain and computes a factorization. The \emph{signal path} (top row) applies two transforms, one banded solve and one banded product to the received samples. We describe each stage of a proposed low-complexity receiver in order below in Fig.~\ref{fig:rx}.
\begin{figure*}[t]\centering
	\begin{tikzpicture}[font=\scriptsize, >=latex, rounded corners=1.5pt,
		blk/.style={draw, rounded corners=1.5pt, align=center, minimum height=9mm, inner xsep=3pt, inner ysep=2.5pt},
		key/.style={blk, fill=blue!7},
		opt/.style={blk, dashed},
		lbl/.style={align=center}]
		\node[lbl] (y) {$\by$\\(DD)};
		\node[blk, right=2.4mm of y] (rot) {rotate: $\bZ^\hmn$,\\then $\bF$ if $X{=}F$};
		\node[key, right=2.4mm of rot] (solve) {solve $\bPsi_X\bz=\by_X$\\bordered substitution};
		\node[opt, right=2.4mm of solve] (pcg) {PCG\\refinement};
		\node[blk, right=2.4mm of pcg] (prod) {$\bH_X^\hmn\bz$ from $\bg_{k'}$\\banded product};
		\node[blk, right=2.4mm of prod] (back) {rotate back: $\bF^\hmn$\\if $X{=}F$, then $\bZ$};
		\node[lbl, right=2.4mm of back] (xh) {$\hat\bx$};
		\draw[->] (y)--(rot);  \draw[->] (rot)--(solve);  \draw[->] (solve)--(pcg);
		\draw[->] (pcg)--(prod);  \draw[->] (prod)--(back);  \draw[->] (back)--(xh);
		\node[key, below=12mm of solve] (fac) {form $\bPsi_X$, then factor\\$\bT=\bL\bL^\hmn$, $\bE$, $\bG$};
		\node[blk, left=2.4mm of fac] (dia) {diagonals $\bg_{k'}$ or $\bff_{l'}$:\\one FFT per tap};
		\node[blk, left=2.4mm of dia] (sel) {marginals:\\choose $X$, band, $D$};
		\node[lbl, left=2.4mm of sel] (h) {$\hat h[k',l']$\\(DD estimate)};
		\draw[->] (h)--(sel);  \draw[->] (sel)--(dia);  \draw[->] (dia)--(fac);
		\draw[->] (fac.north)--node[right, font=\scriptsize, inner sep=2pt]{$\bL,\bE,\bG$}(solve.south);
		\node[font=\scriptsize\itshape, above=1.2mm of prod.north] {Signal path: once per frame};
		\node[font=\scriptsize\itshape, below=1.2mm of dia.south] {Channel path: once per channel estimate};
	\end{tikzpicture}
	\caption{Architecture of the proposed receiver. The rotated observation is $\by_X=\bZ^\hmn\by$ in the TD branch and $\by_X=\bF\bZ^\hmn\by$ in the FD branch, and $\bPsi_X$ is the LMMSE system matrix in that domain. Dashed boxes are used only when the filter tails are too heavy for the retained band. The channel path selects the domain and factorizes the system matrix, and is repeated only when the channel estimate changes; the signal path then costs $\mathcal O(MN(\log N+D))$ per frame. }\label{fig:rx}
\end{figure*}
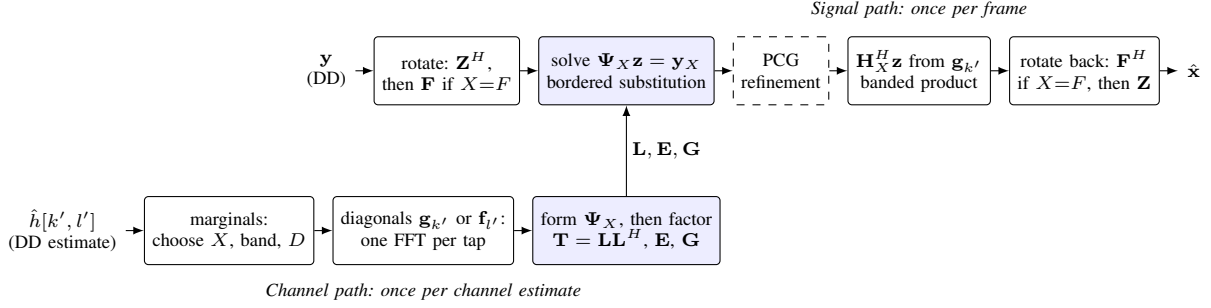
\subsection{Exact Rotation of the LMMSE Problem}\label{sec:lmmse}
The first step shows that the DD-LMMSE estimator \eqref{eq:lmmse} can be computed in TD or FD \emph{without approximation}. This requires the noise covariance in the new domain.

The noise part of \eqref{eq:yt} is $\nu(t)=T^{-1/2}W_2^*(t)\int w_1^*(u-t)n(u)\,du$, whose correlation is
\[
\mathbb E[\nu(t)\nu^*(t')]=\frac{N_0}{T}W_2^*(t)W_2(t')\int w_1^*(u-t)\,w_1(u-t')\,du.
\]
Substituting $u=\frac{t+t'}{2}+v$ in \eqref{eq:chi} with $\nu=0$ shows that the integral equals $\chi_0(t-t')$. Periodizing as in \eqref{eq:yT} and rescaling by $T$ gives $\bn=\bZ\bn_T$ with $\mathbb E[\bn_T\bn_T^\hmn]=N_0\bR_T$ and
\begin{equation}\label{eq:RT}
	\bR_T[t,t']=\sum_{p,p'\in\Z}\hat W^*(t_p)\,\hat W(t'_{p'})\,\chi_0(t_p-t'_{p'}),
\end{equation}
so that $\bR=\bZ\bR_T\bZ^\hmn$. Since $\chi_0$ decays with the pulse autocorrelation, $\bR_T$ is cyclically banded, with a width set by the support of $\chi_0$; for the sinc pulse with the rectangular window, $\chi_0(d)=\mathrm{sinc}(Bd)$ vanishes at nonzero multiples of $1/B$ and $\hat W=1$ on the frame, so $\bR_T=\bI$.

Substituting $\bH=\bZ\bH_T\bZ^\hmn$ and $\bR=\bZ\bR_T\bZ^\hmn$ into \eqref{eq:lmmse} and using $\bZ^\hmn\bZ=\bI$ gives
\begin{align}
	\hat\bx&=\bZ\bH_T^\hmn\bZ^\hmn\left(\bZ(\bH_T\bH_T^\hmn+N_0\bR_T)\bZ^\hmn\right)^{-1}\by\nonumber\\
	&=\bZ\,\bH_T^\hmn\left(\bH_T\bH_T^\hmn+N_0\bR_T\right)^{-1}\bZ^\hmn\by.\label{eq:lmmseT}
\end{align}
Applying the same argument with $\bH_F=\bF\bH_T\bF^\hmn$ and $\bR_F=\bF\bR_T\bF^\hmn$,
\begin{equation}\label{eq:lmmseF}
	\hat\bx=\bZ\bF^\hmn\,\bH_F^\hmn\left(\bH_F\bH_F^\hmn+N_0\bR_F\right)^{-1}\bF\bZ^\hmn\by.
\end{equation}
Both \eqref{eq:lmmseT} and \eqref{eq:lmmseF} are \emph{identical} to \eqref{eq:lmmse}. Writing $X\in\{T,F\}$, $\by_X$ for the rotated observation ($\by_T=\bZ^\hmn\by$, $\by_F=\bF\bZ^\hmn\by$) and
\begin{equation}\label{eq:psiX}
	\bPsi_X\triangleq\bH_X\bH_X^\hmn+N_0\bR_X,
\end{equation}
the whole receiver reduces to three operations: solve $\bPsi_X\bz=\by_X$, form $\bH_X^\hmn\bz$, and transform back. All savings come from the structure of $\bPsi_X$.

\subsection{Domain and Band Selection}\label{sec:trunc}
Corollary~\ref{cor:band} states that the TD band follows the delay extent of $h$ and the FD band its Doppler extent. Which is smaller depends on the channel, so the receiver decides per realization.

Let $e_{k'}\triangleq\|\bg_{k'}\|^2$ (TD) or $e_{l'}\triangleq\|\bff_{l'}\|^2$ (FD). For each $D$, the out-of-band fraction is
\begin{equation}\label{eq:eps}
	\varepsilon_X(D)=1-\max_{u}\frac{\sum_{j=0}^{D-1}e_{u+j}}{\sum_{j}e_j},
\end{equation}
where the maximum runs over cyclic window starts $u$. Its maximizer defines the retained set $\mathcal K$ (TD) or $\mathcal L$ (FD) of $D$ consecutive diagonals. By Lemma~\ref{lem:energy}, $e_{k'}$ and $e_{l'}$ are simply the delay and Doppler marginals of $|h[k',l']|^2$, scaled by $n$. Evaluating \eqref{eq:eps} for all $D$ therefore takes $\mathcal O(n)$ operations with prefix sums, and no matrix is formed. The receiver sets
\begin{equation}\label{eq:rule}
	\!\!\!D_X=\min\{D:\varepsilon_X(D)\le\varepsilon_0\},\; X^\star=\arg\min_{X\in\{T,F\}}D_X,
\end{equation}
because the cost of the later stages grows as $D_X^2$ (Section~\ref{sec:cplx}).

\emph{Physical interpretation.} From \eqref{eq:HTphys}, the delay extent of $h$ is the delay spread $B\tau_{\max}$ (in bins) plus the tails of $\chi$ along delay. The Doppler extent is the Doppler spread $2T\nu_{\max}$ (in bins) plus the tails of the window spectrum. Hence
\begin{equation}\label{eq:Dheur}
	D_T\approx B\tau_{\max}+1+2\kappa_\tau,\qquad D_F\approx 2T\nu_{\max}+1+2\kappa_\nu,
\end{equation}
where $\kappa_\tau$ and $\kappa_\nu$ are the one-sided tail lengths of the filter responses. TD is preferable when the Doppler spread in bins exceeds the delay spread in bins, i.e., in high-mobility links with short delay spread. FD is preferable in the opposite case. Section~\ref{sec:results} confirms \eqref{eq:Dheur} for both regimes.

\subsection{Channel Diagonals and the System Matrix}
Given the domain $X$ and the retained set ($\mathcal K$ in TD, $\mathcal L$ in FD), the receiver computes only the retained diagonals of $\bH_X$ from \eqref{eq:diags}. Each $\bg_{k'}$ is the $n$-point IFFT of the Doppler profile $h[k',\cdot]$, modulated by $\omega^{-k'l'}$ inside the sum. Each $\bff_{l'}$ is the $n$-point FFT of the delay profile $h[\cdot,l']$. This costs $D\,\tfrac n2\log_2n$ CMs, and neither the path parameters nor a dense matrix are needed.
The diagonals of $\bPsi_X$ follow from those of $\bH_X$. Write $\psi_c[t]\triangleq\bPsi_X[t,t-c]$. Since $(\bH_X\bH_X^\hmn)[t,t-c]=\sum_m\bH_X[t,m]\bH_X^*[t-c,m]$ and $\bH_X[t,m]\neq0$ only for $m=t-k'$ with $k'\in\mathcal K$, we need $\bH_X[t-c,t-k']=g_{k''}[t-c]$ with $(t-c)-k''=t-k'$, i.e., $k''=k'-c$. Therefore
\begin{equation}\label{eq:psidiag}
	\psi_c[t]=\!\!\sum_{\substack{k'\in\mathcal K\\ k'-c\in\mathcal K}}\!\!g_{k'}[t]\,g^*_{k'-c}[t-c]+N_0\,r_c[t],\quad|c|\le w,
\end{equation}
where $r_c$ are the diagonals of $\bR_X$. The cyclic half-bandwidth is $w=\max(D-1,w_R)$, where $w_R$ is the half-bandwidth of $\bR_X$. Only $c\ge0$ needs to be computed, since $\psi_{-c}[t]=\psi_c^*[t+c]$. The cost is $nD(D+1)/2$ CMs. Finally, $\bPsi_X$ is Hermitian positive definite (HPD) whenever $\bR_X\succ0$, as for the filters of Section~\ref{sec:model}: $\bH_X\bH_X^\hmn\succeq0$ and $N_0\bR_X\succ0$.

\subsection{Bordered Banded Cholesky Factorization}\label{sec:chol}
If $\bPsi_X$ were strictly banded, a banded Cholesky factorization would solve the system in $\mathcal O(nw^2)$ \cite{golub}. However, the quasi-periodicity of Zak-OTFS makes the band \emph{cyclic}: diagonal $c$ wraps around, so entries appear in the top-right and bottom-left corners (Fig.~\ref{fig:part}). Viewed as an ordinary banded matrix, a cyclically banded matrix has half-bandwidth close to $n$, and a banded factorization gains nothing. Prior work avoids the corners by approximating them \cite{mohammedFD} or by removing them through null-space mounting at a cost in degrees of freedom \cite{mattuFD}. We keep them and treat them exactly. The idea is that all corner entries involve the last $w$ indices. Moving these indices into a small \emph{border} leaves a leading block that is strictly banded (Fig.~\ref{fig:part}). The border couples to only $w$ unknowns, and its effect is absorbed by a $w\times w$ Schur complement.
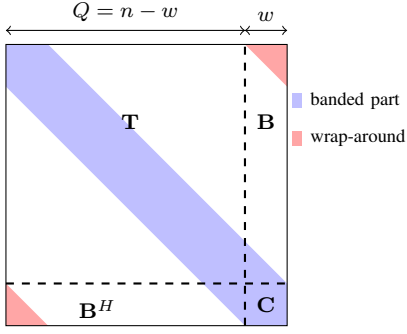
\begin{figure}[t]\centering
	\begin{tikzpicture}[scale=0.62, font=\footnotesize]
		\def\n{6} \def\w{0.9} \def\q{5.1}
		\fill[blue!25] (0,\n) -- (\w,\n) -- (\n,\w) -- (\n,0) -- (\n-\w,0) -- (0,\n-\w) -- cycle;
		\fill[red!35] (\n-\w,\n) -- (\n,\n) -- (\n,\n-\w) -- cycle;
		\fill[red!35] (0,0) -- (\w,0) -- (0,\w) -- cycle;
		\draw (0,0) rectangle (\n,\n);
		\draw[dashed, thick] (\q,0) -- (\q,\n); \draw[dashed, thick] (0,\n-\q) -- (\n,\n-\q);
		\node at (\q/2+0.1,\n-\q/2+0.9) {$\bT$};
		\node at (\q+\w/2,\n-\q/2+0.9) {$\bB$};
		\node at (\q/2-0.6,0.35) {$\bB^\hmn$};
		\node at (\q+\w/2,0.45) {$\bC$};
		\draw[<->] (0,\n+0.3) -- node[above]{$Q=n-w$} (\q,\n+0.3);
		\draw[<->] (\q,\n+0.3) -- node[above]{$w$} (\n,\n+0.3);
		\node[align=left, anchor=west] at (\n+0.3,\n-1.2) {\scriptsize banded part};
		\fill[blue!25] (\n+0.1,\n-1.35) rectangle ++(0.2,0.3);
		\node[align=left, anchor=west] at (\n+0.3,\n-2.0) {\scriptsize wrap-around};
		\fill[red!35] (\n+0.1,\n-2.15) rectangle ++(0.2,0.3);
	\end{tikzpicture}
	\caption{Cyclically banded $\bPsi_X$ and the partition \eqref{eq:part}. The wrap-around entries (red) lie only in the border $\bB$, $\bB^\hmn$, $\bC$; the leading block $\bT$ is strictly banded.}\label{fig:part}
\end{figure}
This is illustrated in Fig.~\ref{fig:part}. The blue band holds the entries within $w$ of the main diagonal, and the red triangles are the wrap-around entries, which sit in the top-right and bottom-left corners. Cutting away the last $w$ rows and columns, along the dashed lines, leaves a leading block that contains no red entry, while every red entry falls in the thin border. 
Formally, let $Q=n-w$ and partition
\begin{equation}\label{eq:part}
	\bPsi_X=\begin{bmatrix}\bT&\bB\\ \bB^\hmn&\bC\end{bmatrix},\qquad \bT\in\mathbb C^{Q\times Q},\;\bC\in\mathbb C^{w\times w}.
\end{equation}
Two index observations confirm what Fig.~\ref{fig:part} suggests. Entry $(t,t')$ of $\bT$, with $0\le t,t'<Q$, is nonzero only if $(t-t')\bmod n\in[-w,w]$; since $|t-t'|\le Q-1=n-w-1$, the only possibility is $t-t'\in[-w,w]$ itself, so \emph{no wrap-around entry lies in} $\bT$, which is therefore strictly banded with half-bandwidth $w$. In $\bB$ (rows $t<Q$, columns $t'\ge Q$), a direct entry requires $t\ge Q-w$ and a wrap-around entry requires $t-t'+n\le w$, that is $t<w$, so the nonzero entries of $\bB$ lie in its first and last $w$ rows.

We now factor $\bPsi_X=\bL_\Psi\bL_\Psi^\hmn$ using the bordered lower-triangular factor
\begin{equation}\label{eq:fac}
	\bL_\Psi=\begin{bmatrix}\bL&\mathbf 0\\ \bE^\hmn&\bG\end{bmatrix}.
\end{equation}
Multiplying out $\bL_\Psi\bL_\Psi^\hmn$ and equating the blocks with \eqref{eq:part} gives
\begin{equation}\label{eq:blocks}
	\bL\bL^\hmn=\bT,\qquad\bL\bE=\bB,\qquad\bG\bG^\hmn=\bC-\bE^\hmn\bE.
\end{equation}
These equations are evaluated sequentially. The first is a banded Cholesky factorization, while the second requires $w$ banded forward substitutions, one for each column of $\bB$. The final equation is a dense Cholesky factorization of the $w\times w$ Schur complement.

Since $\bT$ is a principal submatrix of the Hermitian positive-definite matrix, it is also positive definite and therefore admits a Cholesky factorization. Moreover, because $\bT$ is strictly banded, its Cholesky factor $\bL$ has the same bandwidth. The Schur complement is 
\begin{align}
	\bC-\bE^\hmn\bE=\bC-\bB^\hmn\bT^{-1}\bB.
\end{align}
is likewise positive definite, so $\bG$ exists. Thus, no pivoting is required \cite{golub}.


Finally, solving
\begin{equation}
	\boldsymbol{\Psi}_{X}\mathbf{z}
	=\mathbf{y}_{X},
	\qquad
	\mathbf{y}_{X}
	=
	\begin{bmatrix}
		\mathbf{r}_{1}\\
		\mathbf{r}_{2}
	\end{bmatrix},
\end{equation}
using
$\boldsymbol{\Psi}_{X}
=\mathbf{L}_{\Psi}\mathbf{L}_{\Psi}^{H}$
requires a forward and a backward substitution. The forward substitution
first solves $\mathbf{L}\mathbf{u}_{1}=\mathbf{r}_{1}$ and then
$\mathbf{G}\mathbf{u}_{2}
=\mathbf{r}_{2}-\mathbf{E}^{H}\mathbf{u}_{1}$.
The backward substitution proceeds in reverse order, solving
$\mathbf{G}^{H}\mathbf{z}_{2}=\mathbf{u}_{2}$ followed by
$\mathbf{L}^{H}\mathbf{z}_{1}
=\mathbf{u}_{1}-\mathbf{E}\mathbf{z}_{2}$.
Thus, the bordered Cholesky factorization yields the exact solution of
the cyclically banded LMMSE system while retaining the wrap-around
terms through the border matrices $\mathbf{E}$ and $\mathbf{G}$.

With $\mathbf{z}$ obtained, the receiver forms the equalized vector
$\mathbf{H}_{X}^{H}\mathbf{z}$ in the selected domain. If $X=F$, the
result is first transformed by $\mathbf{F}^{H}$ and then mapped to the
DD domain through the inverse DZT. The resulting estimate is finally
sliced to the signal constellation. The complete receiver procedure is
summarized in Algorithm~\ref{alg:bbc}.

\begin{algorithm}[t]
	\caption{Proposed LMMSE receiver for Zak-OTFS}\label{alg:bbc}
	\begin{algorithmic}[1]
		\Require DD samples $\by$; effective channel $h[k',l']$ on its support; $N_0$; tolerance $\varepsilon_0$; PCG steps $K\ge0$
		\Statex \textit{Channel path}
		\State Compute the marginals \eqref{eq:energy} and $\varepsilon_T(D)$, $\varepsilon_F(D)$ \eqref{eq:eps}; choose $X$, the retained set, and $D$ by \eqref{eq:rule}
		\State Compute the retained diagonals $\{\bg_{k'}\}_{k'\in\mathcal K}$ or $\{\bff_{l'}\}_{l'\in\mathcal L}$ via \eqref{eq:diags} with FFTs
		\State Form $\psi_c$ by \eqref{eq:psidiag}; if $\bar\varepsilon$ is not negligible, add $\bar\varepsilon$ to $\psi_0$ \eqref{eq:psiD}
		\State Factor $\bT=\bL\bL^\hmn$; compute $\bE=\bL^{-1}\bB$ and $\bG\bG^\hmn=\bC-\bE^\hmn\bE$ \eqref{eq:blocks}
		\Statex \textit{Signal path}
		\State $\by_X\gets\bZ^\hmn\by$ ($M$ IFFTs of size $N$); if $X=F$, $\by_X\gets\bF\by_X$
		\State Solve $\bPsi_X\bz=\by_X$ using the bordered Cholesky factors
		\If{$K>0$}
		\State Refine $\bz$ using $K$ PCG iterations with the bordered Cholesky factorization as preconditioner
		\EndIf
		\State Transform back to the DD domain: $\hat{\bx}\gets\bZ\bH_X^\hmn\bz$; if $X=F$, apply $\bF^\hmn$ before $\bZ$; slice $\hat{\bx}$ to the constellation
	\end{algorithmic}
\end{algorithm}
\subsection{Refinement for Heavy-Tailed Filters}\label{sec:pcg}
With rapidly decaying filters (e.g., Gaussian), $\varepsilon_X(D)$ falls super-exponentially once $D$ covers the physical spread, and the receiver above is exact for practical purposes. With slowly decaying filters (e.g., sinc, whose $\chi$ decays as $1/|d|$), any affordable band leaves a residual. We treat this residual in two steps of increasing cost.

\emph{Mismatch-aware regularization (no extra cost).} Write $\bH_X=\bH_D+\bE_D$, where $\bH_D$ keeps the retained diagonals. The two terms have disjoint supports, so the cross term $\bH_D\bE_D^\hmn$ has zero diagonal and
\[
\mathrm{diag}(\bH_X\bH_X^\hmn)=\mathrm{diag}(\bH_D\bH_D^\hmn)+\mathrm{diag}(\bE_D\bE_D^\hmn).
\]
Replacing the neglected terms by their average diagonal $\bar\varepsilon\triangleq\|\bE_D\|_F^2/n$ yields
\begin{equation}\label{eq:psiD}
	\bPsi_D=\bH_D\bH_D^\hmn+N_0\bR_X+\bar\varepsilon\,\bI,
\end{equation}
which preserves $\mathrm{tr}(\bPsi_X)$ and the band. The residual $\bE_D\bs$ acts as interference of average power $\bar\varepsilon$, so to first order the post-equalization signal-to-interference-plus-noise ratio (SINR) is limited to about $1/\bar\varepsilon$. By Lemma~\ref{lem:energy}, $\bar\varepsilon=\varepsilon_X(D)\,\|\bH_X\|_F^2/n$ is known from the marginals. This step softens the error floor but cannot remove it.

\emph{Banded-preconditioned CG (a few iterations).} To remove the floor, we solve the \emph{exact} system $\bPsi_X\bz=\by_X$ by preconditioned CG \cite{saad}. The factorization of $\bPsi_D$ from Section~\ref{sec:chol} serves as the preconditioner, applied at $\mathcal O(nw)$ per step. PCG converges as
\[
\|\bz_K-\bz^\star\|_{\bPsi}\le2\left(\frac{\sqrt\kappa-1}{\sqrt\kappa+1}\right)^K\|\bz_0-\bz^\star\|_{\bPsi},
\]
with $\kappa=\kappa(\bPsi_D^{-1}\bPsi_X)$. Since $\bPsi_D^{-1}\bPsi_X=\bI+\bPsi_D^{-1}(\bPsi_X-\bPsi_D)$ and the perturbation is small in norm, the eigenvalues cluster near one and $\kappa$ is small. The narrow band captures most of the interference, and CG corrects the remainder.

PCG needs products with the exact $\bH_X$. They can be computed from the full estimated support of $h$ as $\bH_X\bv=\sum_{k'}\bg_{k'}\odot\bPi^{k'}\bv$. Alternatively, when the path parameters $(h_i,\tau_i,\nu_i)$ are available from a parametric channel estimate, they can be computed with FFTs at a cost independent of the channel spread. For the sinc pulse, \eqref{eq:HTphys} with $p=q=0$ gives
\begin{align}
	\bH_T[t,t']&=\sum_i h_i\,e^{-j\pi\nu_i\tau_i}\,e^{j\pi\nu_i\frac tB}\,\kappa_i[t-t']\,e^{j\pi\nu_i\frac{t'}{B}},\nonumber\\
	\bH_T&=\sum_i h_i\,e^{-j\pi\nu_i\tau_i}\,\bPhi_i\bK_i\bPhi_i,\label{eq:toep}
\end{align}
where $\kappa_i[d]=\chi_{\nu_i}(d/B-\tau_i)$, $\bPhi_i=\mathrm{diag}(e^{j\pi\nu_it/B})$ and $\bK_i$ is Toeplitz. Each Toeplitz product is a linear convolution computed with FFTs of length $L\ge2n-1$. One iteration, with one product by $\bH_T$ and one by $\bH_T^\hmn$, costs $2P(L\log_2L+L+2n)$ CMs.

\subsection{Computational Complexity}\label{sec:cplx}
We show in Table~\ref{tab:cplx} the number of CMs required by each operation of the receiver. Band selection needs only $\mathcal O(n)$ operations, since by Lemma~\ref{lem:energy} it works on the two marginals of the channel estimate. Computing the retained diagonals of $\bH_X$ from \eqref{eq:diags} costs $D\frac n2\log_2n$, one length-$n$ FFT per retained tap, and assembling the diagonals of $\bPsi_X$ through \eqref{eq:psidiag} costs $nD(D+1)/2$, the dominant term for wide bands. The factorization \eqref{eq:blocks} adds about $nw(w+3)/2$ for the banded Cholesky factor of $\bT$, $\tfrac32nw^2$ for $\bE$ and the Schur complement, and $w^3/6$ for $\bG$, while the signal path costs $n\log_2N$ for the two transforms ($n\log_2(n)$ more on the FD branch), $4nw$ for the substitutions and $nD$ for $\bH_X^\hmn\bz$.
The overall complexity is $\mathcal O\big(MN(\log N+D^2)\big)$, linear in the frame size and fixed in advance for a given band width, so the receiver meets a latency budget that an iterative solver in \cite{mattuFD} cannot guarantee. The bordered steps that make the wrap-around exact contribute only $\mathcal O(nw^2)$, the same order as the banded factorization itself, so treating the quasi-periodicity exactly is essentially free. In practice the channel-dependent work is repeated only when the estimate changes, so a receiver that equalizes several frames per estimate pays only $\mathcal O(n(\log N+D))$ per frame.
Compared with \cite{tiwari}, the partition \eqref{eq:part} uses Cholesky rather than LU, which is possible because $\bPsi_X$ is Hermitian positive definite. This halves both the arithmetic and the storage, and it removes the need for pivoting. It also carries \emph{both} wrap-around corners in the border, rather than one, because the Zak-OTFS pulses are non-causal and place taps on either side of the diagonal. Compared with \cite{mohammedFD,mattuFD}, the wrap-around is carried exactly by the border, so all $MN$ DD positions remain available to data and the solution is direct.

Table~\ref{tab:prior} summarizes the key differences between the proposed
receiver and the FD equalizers in \cite{mohammedFD,mattuFD} and the TD
equalizer in \cite{zhang2025ogd}. Building on the use of a unitarily
equivalent banded domain for low-complexity equalization, the proposed
receiver selects the processing domain adaptively, preserves the
quasi-periodic wrap-around exactly, and replaces iterative equalization
with a direct bordered-Cholesky solution.
\begin{table}[t]
	\centering
	\caption{Complex Multiplications of the Proposed Receiver ($n=MN$, $w=D-1$)}\label{tab:cplx}
	\setlength{\tabcolsep}{3pt}\small
	\begin{tabular}{@{}ll@{}}
		\toprule
		Operation & CMs\\\midrule
		Diagonals from $h$ \eqref{eq:diags} & $D\,\tfrac n2\log_2 n$\\
		DZT and inverse DZT & $n\log_2N$\\
		FFT and IFFT (FD only) & $n\log_2(n)$\\
		Diagonals of $\bPsi_X$ \eqref{eq:psidiag} & $n\,D(D+1)/2$\\
		Banded Cholesky of $\bT$ & $n\,w(w+3)/2$\\
		Border: $\bE$, Schur, $\bG$ & $\tfrac32n w^2+\tfrac{w^3}{6}$\\
		Substitutions and $\bH_X^\hmn\bz$ & $n(4w+D)$\\\midrule
		Total (order) & $\mathcal O\big(n(\log N+D^2)\big)$\\
		Direct DD-LMMSE \eqref{eq:lmmse} & $\approx\tfrac23(n)^3$\\
		PCG step (sinc, FFT) & $2P(L\log_2\!L{+}L{+}2n){+}4nw$\\\bottomrule
	\end{tabular}
\end{table}

\begin{table*}[t]
    \centering
    \caption{Comparison with prior Zak-OTFS equalizers.}
    \label{tab:prior}
    \setlength{\tabcolsep}{6pt}
    \begin{tabular}{@{}p{0.14\textwidth}
                    p{0.25\textwidth}
                    p{0.23\textwidth}
                    p{0.30\textwidth}@{}}
        \toprule
        & \cite{mohammedFD}/\cite{mattuFD}
        & \cite{zhang2025ogd}
        & This work \\ \midrule

        Domain
        & FD
        & TD
        & TD or FD, selected \\

        Structure
        & Doppler-banded
        & Channel sparsity
        & TD: delay-banded; \\& & &FD: Doppler-banded \\

        Wrap-around
        & Extended vector / null-space mounting
        & ---
        & Exact bordered factorization \\

        Solver
        & Banded inverse / CG
        & Offset gradient descent
        & Direct Cholesky (+PCG) \\

        Cost
        & $\mathcal{O}(bn^2)$ / $\mathcal{O}(Kbn)$
        & Iterative, $K$ steps
        & $\mathcal{O}\!\left(n(\log N+D^2)\right)$, fixed \\

        Noise
        & White
        & ---
        & Colored $\bR_X$ \\

        \bottomrule
    \end{tabular}
\end{table*}
\section{Numerical Results}\label{sec:results}
This section evaluates the proposed receiver against direct DD-LMMSE in two regimes: a rapidly decaying filter, for which a narrow band suffices, and a slowly decaying one, for which it does not. 
Throughout, a frame carries $M=31$ delay bins and $N=37$ Doppler bins ($n=1147$ symbols), with a Doppler period $\nu_p=30$~kHz, a bandwidth $B=930$~kHz and a frame duration $T=1.233$~ms, as in \cite{mohammedFD,mattuFD}, so that the complexity figures below are directly comparable with the frequency-domain receivers reported there.
We use the International Telecommunication Union (ITU) vehicular-A (Veh-A, $\tau_{\max}=2.51~\mu$s) and vehicular-B (Veh-B, $\tau_{\max}=20~\mu$s $<\tau_p$) profiles \cite{itu}, with Jakes Doppler $\nu_i=\nu_{\max}\cos\theta_i$, $\theta_i\sim\mathcal U[0,2\pi)$, and independent Rayleigh gains. Two filters are compared: the Gaussian filter of \cite{mohammedFD} ($\alpha_g=\beta_g=1.584$), for which $\bR_T\neq\bI$, and the sinc pulse with a rectangular window, for which $\bR_T=\bI$. Transmission is uncoded 4-ary quadrature amplitude modulation (4-QAM) with perfect CSI. The SNR is $\rho=(E/N_0)/\big((1+\tau_{\max}/T)(1+2\nu_{\max}/B)\big)$ for the Gaussian filter \cite{mohammedFD} and $E/N_0$ for sinc. These two filters bracket the behaviour of interest: the Gaussian decays fastest and the sinc slowest, while the Gaussian-sinc filter of \cite{gaussiansinc} lies between them and would inherit the intermediate band width. The TD channel and noise are generated exactly from \eqref{eq:HTphys} and \eqref{eq:RT} with $|p|,|q|\le5$. The reference DD-LMMSE applies \eqref{eq:lmmse} to the dense matrix $\bZ\bH_T\bZ^\hmn$.


\subsection{Band Widths and Domain Selection}
Fig.~\ref{fig:structure} showed the band structure of a single small frame; Fig.~\ref{fig:band} and Table~\ref{tab:D} now quantify it over many realizations of the full-size frame and averaged over 20 channels.
\begin{figure}[t]\centering
	\includegraphics[width=\columnwidth]{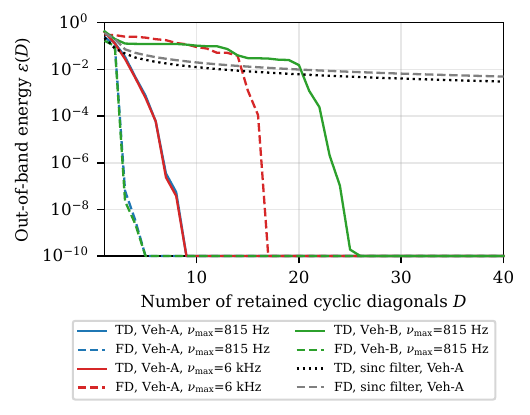}
	\caption{Out-of-band energy \eqref{eq:eps} versus number of retained diagonals, TD (solid) and FD (dashed), averaged over 20 channels; Gaussian filter unless stated.}\label{fig:band}
\end{figure}
\begin{table}[t]\centering
	\caption{Diagonals for $\varepsilon_0=10^{-4}$ and CMs at $M=31,N=37$ (Direct DD-LMMSE: $1.0\times10^9$ CMs)}\label{tab:D}
	\setlength{\tabcolsep}{3.5pt}
	\begin{tabular}{@{}lcccc@{}}\toprule
		Channel & $D_T$ & $D_F$ & Domain & CMs (reduction)\\\midrule
		Veh-A, 815 Hz & 6 & 3 & FD & $5.0\times10^4$ ($2.0\times10^4\times$)\\
		Veh-A, 6 kHz & 6 & 17 & TD & $1.3\times10^5$ ($8.0\times10^3\times$)\\
		Veh-B, 815 Hz & 23 & 3 & FD & $5.0\times10^4$ ($2.0\times10^4\times$)\\\bottomrule
	\end{tabular}
\end{table}
We observe that the TD curves for Veh-A at 815~Hz and at 6~kHz coincide: the TD band does not depend on Doppler. Likewise, the FD curves of Veh-A and Veh-B at 815~Hz coincide. The measured widths agree with \eqref{eq:Dheur}. For Veh-A at 6~kHz ($2T\nu_{\max}=14.8$, $B\tau_{\max}=2.3$), TD needs 6 diagonals and FD 17, so TD reduces the $D^2$-dependent cost about $7\times$. For Veh-B ($B\tau_{\max}=18.6$), FD needs 3 and TD 23. With the sinc filter, $\varepsilon_T(D)$ decays slowly: $D=12$ reaches $10^{-2}$, but even $D=60$ does not reach $10^{-3}$. This matches the observation of \cite{mattuFD} for such filters and motivates Sec.~\ref{sec:pcg}.
\begin{figure}[t]\centering
\includegraphics[width=\columnwidth]{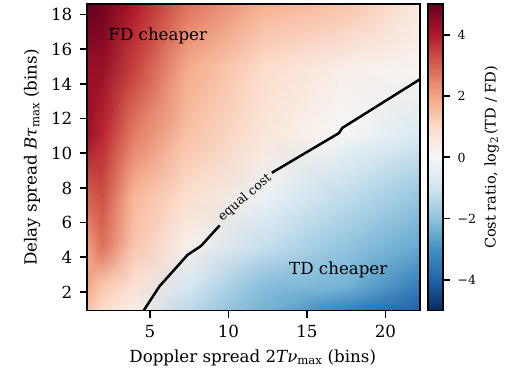}
\caption{Ratio of the cost of TD equalization to that of FD equalization over the spread plane, at $\varepsilon_0=10^{-4}$, averaged over four channels per point. The contour marks equal cost.}\label{fig:map}
\end{figure}

Fig.~\ref{fig:map} generalizes these three operating points to the whole plane of delay and Doppler spreads. For each pair $(B\tau_{\max},2T\nu_{\max})$ we draw channels with that spread, choose $D_T$ and $D_F$ by \eqref{eq:rule}, and plot the ratio of the resulting costs. The contour of equal cost follows $2T\nu_{\max}\approx1.5\,B\tau_{\max}$, as \eqref{eq:Dheur} predicts once the filter tails are included, and it separates two regimes: time is cheaper in the high-mobility corner, frequency in the long-delay corner. The penalty for working in the wrong domain reaches a factor of $42$ at the edges of the plane and exceeds an order of magnitude over much of it, so the choice is worth making per channel rather than once at design time.
\subsection{Bit Error Rate}
Figs.~\ref{fig:bergauss} and \ref{fig:bersinc} show the bit error rate (BER) versus SNR for the two filters.
With the Gaussian filter and Veh-A at $\nu_{\max}=6$~kHz, where TD is the preferred domain, the proposed receiver with $D=6$ coincides with DD-LMMSE at every SNR, both reaching $2.03\times10^{-3}$ at 27~dB, and so does the FD receiver with $D=17$. Narrower windows floor, as anticipated in Section~\ref{sec:trunc}: TD with $D=4$ reaches $6.1\times10^{-3}$ at 27~dB, and FD with $D=11$ and $D=6$ level off at $2.5\times10^{-2}$ and $5.3\times10^{-2}$. The TD receiver therefore attains exact LMMSE performance with $1.3\times10^5$ CMs, against $9.0\times10^5$ in FD and $1.0\times10^9$ for direct DD-LMMSE.

The sinc filter exposes the heavy-tailed case. Truncating to $D=8$, which leaves $\varepsilon_T\approx10^{-2}$, creates a floor near $3.5\times10^{-3}$; the mismatch-aware regularization \eqref{eq:psiD} lowers it to $1.4\times10^{-3}$ at no extra cost, but cannot remove it. Refining with PCG on the exact operator closes the gap, and $K=4$ iterations already follow DD-LMMSE, giving $3.3\times10^{-4}$ against $2.6\times10^{-4}$ at 21~dB and $6.4\times10^{-5}$ against $3.9\times10^{-5}$ at 24~dB. Unpreconditioned CG needs $K=16$ to match and still floors at $1.7\times10^{-3}$ with $K=4$. Counting one product with $\bH_T$ and one with $\bH_T^\hmn$ per iteration, at $6.7\times10^5$ CMs for $P=6$ paths, banded PCG with $K=4$ costs $3.0\times10^6$ CMs in total, against $1.1\times10^7$ for CG with $K=16$ and $1.0\times10^9$ for direct DD-LMMSE.
\begin{figure}[t]\centering
	\includegraphics[width=\columnwidth]{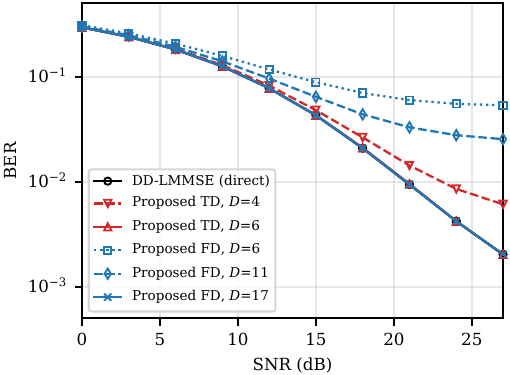}
	\caption{BER, Gaussian filter, Veh-A, $\nu_{\max}=6$~kHz (120 channels, $2.75\times10^5$ bits per point).}\label{fig:bergauss}
\end{figure}
\begin{figure}[t]\centering
	\includegraphics[width=\columnwidth]{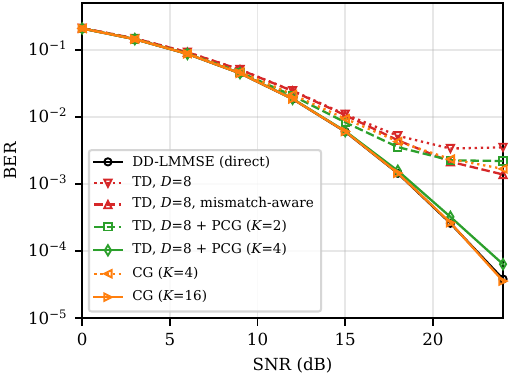}
	\caption{BER, sinc filter, Veh-A, $\nu_{\max}=815$~Hz (170 channels, $3.9\times10^5$ bits per point).}\label{fig:bersinc}
\end{figure}

\subsection{Convergence and Complexity}
\begin{figure}[t]\centering
	\includegraphics[width=\columnwidth]{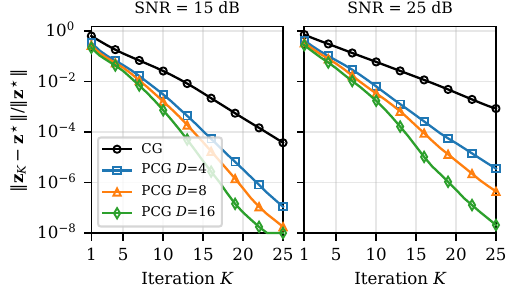}
	\caption{Relative error of CG and banded-preconditioned CG with respect to the exact solution, sinc filter, Veh-A, $\nu_{\max}=815$~Hz (geometric mean over 15 channels).}\label{fig:conv}
\end{figure}
\begin{figure}[t]\centering
	\includegraphics[width=\columnwidth]{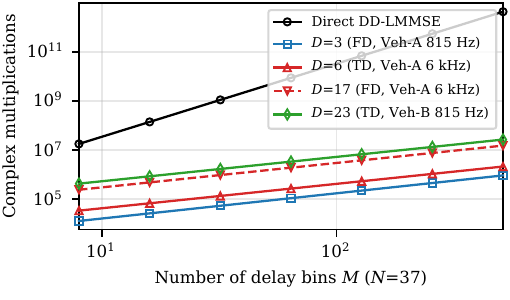}
	\caption{Complex multiplications versus $M$ ($N=37$), from Table~\ref{tab:cplx}.}\label{fig:cplx}
\end{figure}
Fig.~\ref{fig:conv} tracks $\|\bz_K-\bz^\star\|/\|\bz^\star\|$. The banded preconditioner reduces the error after four iterations by a factor of 3--4 relative to CG. It reaches $10^{-3}$ in 10--15 iterations, whereas CG needs 18 at 15~dB and more than 25 at 25~dB. Detection needs much less accuracy: four PCG iterations suffice in Fig.~\ref{fig:bersinc}.

Fig.~\ref{fig:cplx} compares the counts of Table~\ref{tab:cplx} with direct DD-LMMSE. At $M=31$, $N=37$ the reduction ranges from about $600\times$ ($D=23$) to $2\times10^4\times$ ($D=3$), and it grows as $(MN)^2$. The choice of domain matters: for Veh-A at 6~kHz, TD ($D=6$) needs $7\times$ fewer CMs than FD ($D=17$).

\section{Conclusion}\label{sec:concl}
We derived the exact discrete-time model of Zak-OTFS and proved a delay--Doppler duality. The channel is exactly a combination of cyclic delay and Doppler operators in time, and dually in frequency. The TD band is set only by the delay spread and the FD band only by the Doppler spread, and the per-diagonal energies are the marginals of the effective DD channel.
These results yield an exact, direct LMMSE equalizer. It selects the cheaper domain, factorizes the cyclically banded system with a bordered banded Cholesky that handles the quasi-periodic wrap-around without rate loss, and refines with banded-preconditioned CG when the filter tails are heavy. The complexity drops from $\mathcal{O}(n^3)$ to
$\mathcal{O}(n(\log N+D^2))$, with no loss for Gaussian filters and a few PCG steps for sinc filters. 
\ifCLASSOPTIONcaptionsoff
\newpage
\fi
\bibliographystyle{IEEEtran}
\bibliography{IEEEabrv,Low_comp_MMSE}
\end{document}